\documentclass{article}
\usepackage[utf8]{inputenc}
\usepackage[total={6in, 8in}]{geometry}
\usepackage{amsmath,amssymb,amsthm,amsrefs,a4wide}
\usepackage{latexsym}
\usepackage{indentfirst}
\usepackage{placeins}
\usepackage{booktabs}
\usepackage{algorithm}
\usepackage{algorithmic}
\usepackage{multirow}
\usepackage{color, xcolor, colortbl}

\usepackage{subfigure}
\usepackage{graphicx,color}
\usepackage{diagbox}

\usepackage{hyperref}

\usepackage{cleveref}
\usepackage{braket}

\Crefname{figure}{Fig.}{Fig.}
\Crefname{equation}{Eq.}{Eq.}

\theoremstyle{plain}
\newtheorem{theorem}{Theorem}

\theoremstyle{definition}

\theoremstyle{remark}
\newtheorem{remark}{Remark}

\DeclareMathOperator{\argmin}{arg min}

\usepackage{qcircuit}

\renewcommand{\d}{\mathrm{d}}

\newcommand{\eps}{\epsilon}

\newcommand{\abs}[1]{\left\lvert#1\right\rvert}

\newcommand{\bbm}{\begin{bmatrix}}
	\newcommand{\ebm}{\end{bmatrix}}

\newcommand{\DD}{\mathcal{D}}

\newcommand{\GG}{\mathcal{G}}
\newcommand{\EE}{\mathcal{E}}

\newcommand{\ii}{\mathrm{i}}

\newcommand{\half}{\frac{1}{2}}

\newcommand{\nup}{n_{\uparrow}}
\newcommand{\nd}{n_{\downarrow}}

\newcommand{\cis}{{c_{i\sigma}}}
\newcommand{\cisd}{c_{i\sigma}^{\dagger}}
\newcommand{\cjs}{{c_{j\sigma}}}

\newcommand{\nis}{{n_{i\sigma}}}
\newcommand{\niu}{{n_{i\uparrow}}}
\renewcommand{\niu}{{n_{i\uparrow}}}
\newcommand{\nid}{{n_{i\downarrow}}}

\newcommand{\cou}{{c_{1\uparrow}}}
\newcommand{\coud}{c_{1\uparrow}^\dagger}
\newcommand{\ctu}{{c_{2\uparrow}}}
\newcommand{\ctud}{c_{2\uparrow}^\dagger}

\newcommand{\cod}{{c_{1\downarrow}}}
\newcommand{\codd}{c_{1\downarrow}^\dagger}
\newcommand{\ctd}{{c_{2\downarrow}}}
\newcommand{\ctdd}{c_{2\downarrow}^\dagger}

\newcommand{\kud}{\ket{\uparrow\downarrow}}

\newcommand{\sqhalf}{\frac{1}{\sqrt{2}}}

\newcommand{\heff}{H_{\mathrm{eff}}}
\newcommand{\expect}{\mathbb{E}}
\newcommand{\cmt}[1]{{#1}}

\usepackage{authblk}
\author[1]{Hongkang Ni\thanks{hongkang@stanford.edu}} 
\author[2]{Haoya Li\thanks{lihaoya@stanford.edu}} 
\author[1,2]{Lexing Ying\thanks{lexing@stanford.edu}}
\affil[1]{Institute for Computational and Mathematical Engineering, Stanford University, Stanford, CA 94305}
\affil[2]{Department of Mathematics, Stanford University, Stanford, CA 94305}
\begin{document}
\title{Quantum Hamiltonian Learning for the Fermi-Hubbard Model}

\date{}

\maketitle
\begin{abstract}
    This work proposes a protocol for Fermionic Hamiltonian learning. For the Hubbard model defined on a bounded-degree graph, the Heisenberg-limited scaling is achieved while allowing for state preparation and measurement errors. To achieve $\eps$-accurate estimation for all parameters, only $\tilde{\mathcal{O}}(\eps^{-1})$ total evolution time is needed, and the constant factor is independent of the system size. Moreover, 
    our method only involves simple one or two-site Fermionic manipulations, which is desirable for experiment implementation. 
\end{abstract}

\section{Introduction}

A fundamental task in quantum physics is to obtain the Hamiltonian of a given system from its (noisy) time evolution and measurements. This problem is traceable at least as far back as the spectroscopy analysis of a two-level system using Ramsey's interferometry \cite{ramsey1956molecular} and has been extensively studied under the name spectroscopy \cites{wineland1992spin, leibfried2004toward}, quantum sensing \cite{degen2017quantum}, quantum process tomography\cites{nielsen2001quantum, schirmer2004experimental, cole2005identifying, altepeter2003ancilla}, and Hamiltonian tomography \cites{wang2015hamiltonian, lapasar2012estimation, schirmer2009two} using various techniques including quantum state tomography, Bayesian analysis, compressed sensing and machine learning. This paper adopts the relatively recent terminology ``Hamiltonian learning'' \cites{innocenti2020supervised, wiebe2014hamiltonian, wang2017experimental, huang2023learning, li2023heisenberg}. Hamiltonian learning has a variety of applications in quantum algorithms and the generic study of quantum systems. For instance, in the context of using analog quantum algorithms for Hamiltonian simulation, Hamiltonian learning algorithms play a crucial role in verifying and calibrating these analog computers \cites{wiebe2014hamiltonian, wang2017experimental, kokail2021quantum}. \cmt{In particular, cold atoms have been employed for the analog simulation of the Fermi-Hubbard model \cite{tarruell2018quantum}, showcasing a potential application in Fermion Hamiltonian learning.}

Unlike the Hamiltonian learning problem of a two-level system \cites{wineland1992spin, schirmer2004experimental, cole2005identifying, cole2006identifying, ralph2012interleaved, schirmer2015ubiquitous}, learning a many-body Hamiltonian \cites{burgarth2011indirect, shabani2011estimation} is essentially more challenging since different sites are arbitrarily coupled, and the exact eigenstates are intractable. A natural strategy is to first reduce the original system to decoupled subsystems consisting of a single site or a pair of sites and then apply single-body or two-body methods. A typical early example named dynamical decoupling (DD) is introduced in \cite{wang2015hamiltonian} for Hamiltonian learning of many-body systems, where carefully designed Pauli pulses are inserted between evolution operators of the system to eliminate certain coupling terms in the Hamiltonian. Similar ideas have also been used in earlier works \cites{viola1998dynamical, viola1999dynamical} to filter out unwanted terms in the Hamiltonian. 

More recently, the authors of \cite{huang2023learning} proposed a novel decoupling method that essentially introduces additional conservation law by inserting random unitaries that correspond to an artificial symmetry group. This procedure is dubbed ``reshaping'' of the Hamiltonian. By combining the reshaping technique with robust phase estimation (see \cites{kimmel2015robust, ni2023low}), the authors are able to solve the Hamiltonian learning problem of many-body spin systems while achieving the Heisenberg limit, which is a fundamental complexity lower bound in quantum metrology (\cmt{see \cite{giovannetti2006quantum, zwierz2010general, zhou2018achieving,higgins2007entanglement}}). A follow-up work \cite{li2023heisenberg} has extended this strategy to Bosonic systems by carefully dealing with the unbounded Bosonic operators. The Heisenberg-limited scaling is also obtained, though an additional system size factor is involved in the error bound. 

To the best of the authors' knowledge, a Hamiltonian learning algorithm for many-body Fermionic systems with Heisenberg-limited scaling is still missing. Though the renowned Jordan-Wigner transformation can be used to convert Fermionic systems to spin systems (see \cite{reiner2016emulating} for an example of one-dimensional Fermi-Hubbard models), the method in \cite{huang2023learning} cannot be applied to Fermionic systems directly by combining with the Jordan-Wigner transformation due to the additional phase factors. The additional phase factors destroy the locality of the Hamiltonian and result in a much more challenging learning problem. In this paper, we aim to solve the Fermionic Hamiltonian learning problem with operations that are intrinsically Fermionic. \cmt{In stark contrast to the Bosonic Hamiltonian learning algorithm \cite{li2023heisenberg}, the error bound obtained in this work does not depend on the size of the system since the Fermion setting does not involve unbounded operators.}

{\cmt{Most recently, during the revision of the current paper, \cite{mirani2024learning} was proposed to address a more general class of fermionic Hubbard Hamiltonians, including complex hopping amplitudes and nonzero chemical potentials.}

\section{Main results}
\subsection{Preliminaries and Notations}
In this paper, we are specifically concerned with the Fermi-Hubbard model with spins. 
Here, we provide a brief review of the basics of the Fock space that describes such a system and the operators on it. \cmt{One can refer to \cite{lindsey2019quantum} or other textbooks on the quantum field theory for a more detailed discussion.} For Fermions with spin on $N$ sites, the Hilbert space that describes a single Fermion is a $2N$-dimensional complex linear space $\mathcal{H}$ with orthogonal basis vectors $\{\ket{\uparrow}_1, \ket{\downarrow}_1, \ket{\uparrow}_2, \ket{\downarrow}_2, \ldots, \ket{\uparrow}_N, \ket{\downarrow}_N \}$, where the subscripts are indices for the sites. The Fock space $\mathcal{F}$ for a many-body system is defined by exterior algebra $\oplus_{n=0}^\infty\bigwedge^{n}(\mathcal{H})$. Here we further define $\bigwedge^{0}(\mathcal{H})$ as the 1-dimensional space spanned by $\ket{-}$ called the ``vacuum state''. It satisfies $\ket{-}\wedge\ket{\alpha} = \ket{\alpha}\wedge\ket{-} = \ket{\alpha}$ for any $\ket{\alpha}\in\mathcal{F}$. Since the dimension of $\mathcal{H}$ is $2N$, we also have $\oplus_{n=0}^\infty\bigwedge^{n}(\mathcal{H})=\oplus_{n=0}^{2N}\bigwedge^{n}(\mathcal{H})$.

We introduce the notation $\kud_i:= \ket{\uparrow}_i\wedge\ket{\downarrow}_i$. We denote by $\cis$ ($\cisd$) the annihilation (creation) operator on site $i$ of spin $\sigma$ ($\sigma\in\{\uparrow,\downarrow\}$). The creation operator $\cisd$ is defined by $\cisd\ket{\alpha} = \ket{\sigma}_i\wedge\ket{\alpha}$ for any $\ket{\alpha}\in\mathcal{F}$. For example, $c_{i\downarrow}^\dagger\ket{\downarrow}_i = \ket{\downarrow}_i\wedge\ket{\downarrow}_i=0$, and $c_{i\downarrow}^\dagger\ket{\uparrow}_i = \ket{\downarrow}_i\wedge\ket{\uparrow}_i=-\kud_i$.
The annihilation operator $\cis$ is the Hermitian transpose of $\cisd$. It is straightforward to show the anti-commutation relations $\{c_{i\sigma}^\dagger, c_{j\sigma'}\} = \delta_{ij}\delta_{\sigma\sigma'}$ and $\{c_{i\sigma}^\dagger, c_{j\sigma'}^\dagger\} = \{c_{i\sigma}, c_{j\sigma'}\}=0$. We further denote by $\nis = \cisd\cis$ the number operator at site $i$.

The Fermi-Hubbard model \cite{Hubbard1963ElectronCI} stands as one of the simplest yet most extensively studied models for Fermionic systems, successfully explaining the superconductive and magnetic effects of solid materials via the interaction of electrons. The Hamiltonian of a spinful Fermi-Hubbard model 
\begin{equation}\label{eq:Hubbard}
	H = -\!\!\!\!\!\!\!\!\sum_{i\sim j, \sigma\in\{\uparrow,\downarrow\}}\!\!\!\!\!\! h_{ij}\cisd\cjs+\sum_{i}\xi_i \niu\nid,
\end{equation}
is an operator on Fork space, where $i\sim j$ denotes that the sites $i$ and $j$ are neighbors. Typically $h_{ij}=h_{ji}$. We also assume $|h_{ij}| \le 1$ and $|\xi_i|\le 1$. Otherwise, we can scale down $H$ by changing the simulation time. The number of sites is denoted by $N$. The goal is to learn the coefficients $h_{ij}$ and $\xi_i$ of the Hamiltonian \eqref{eq:Hubbard} for $1\le i,j\le N$. The following notations are used for the products of operators:
\begin{equation*}
    \prod_{1 \leq l \leq L}^{\rightarrow} O_l = O_1 O_2\cdots O_L,\quad \prod_{1 \leq l \leq L}^{\leftarrow} O_l = O_L\cdots O_2 O_1.
\end{equation*}

An important ingredient of our approach is the robust phase estimation (RPE) routine. RPE is a technique that gives approximate values of phase from signals with $\mathcal{O}(1)$ error and achieves the Heisenberg limit scaling \cites{kimmel2015robust, ni2023low}, using the idea of iterative refinement for the estimation of the phase from noisy signals. One can refer to \Cref{thm:rpe} and the remark following the theorem for more details. 

\subsection{Main algorithm and theoretical results}
Assume that the Fermionic sites are located on a graph $G = (V, E)$. Each vertex $i\in V$ represents a Fermionic site, and the interaction between two Fermions exists if and only if two edges join them in graph $G$. Therefore, the objective Hamiltonian is 
\begin{equation}\label{eq:hamiltonian}
	H = -\!\!\!\!\!\!\!\!\sum_{(i, j)\in E, \sigma\in\{\uparrow,\downarrow\}}\!\!\!\!\!\! h_{ij}\cisd\cjs+\sum_{i}\xi_i \niu\nid,\quad (h_{ij} = h_{ji}).
\end{equation}
As mentioned earlier, our main strategy is to divide the edge set $E$ into different colors and use the reshaping technique to reduce the Hamiltonian $H$ into an effective Hamiltonian $(\heff)_c$ that only contains the interactions within a certain color $c$, then learn the parameters for the pairs of sites belonging to this color using single-site and two-site methods since they are decoupled. An outline for the algorithm is provided below.

\begin{algorithm}[ht]
	\caption{Fermionic Hamiltonian learning for Hubbard models (outline)}
	\label{alg:outline}
	\begin{algorithmic}[1]
		\STATE{\textbf{Input:} The graph $G = (V,E)$ for the Fermi-Hubbard Hamiltonian $H$ of the form \Cref{eq:hamiltonian}.}
		\STATE{Generate the coloring partition of $E$ with $\chi$ colors (see \Cref{sec:many-site} for details).}
		\FOR{$c = 1,2,\ldots, \chi$}
            \STATE{Prepare initial state $\ket{\Phi_c}$ based on the chosen color $c$ (see \Cref{sec:many-site} for details).}
		\STATE{Reshape the Hamiltonian to get a state that approximates $e^{\ii {(\heff)}_c t_j}\ket{\Phi_c}$ for a set of different simulation time $\{t_j\}_{j=0}^J$ with $\mathcal{O}(1)$ reshaping error (see \Cref{sec:reshaping} for more details).}
            \STATE{Perform measurements using observable $O_c$ that depends on the chosen color $c$ (see \Cref{sec:many-site} for details).} 
            \STATE{Repeat the process sufficiently many times and extract signals from the measurements.}
            \STATE{Apply RPE to get estimations for the parameters corresponding to sites in the considered color $c$ (see \Cref{sec:single-site} and \Cref{sec:two-site} for more details).}
		\ENDFOR
		\STATE{\textbf{Output:} $\hat{h}_{ij}$ and $\hat{\xi_i}$. }
	\end{algorithmic}
\end{algorithm}

The main theoretical result of this work is the following theorem:
\begin{theorem}
(Informal) Assume $H$ is a Hamiltonian in the form \eqref{eq:hamiltonian}. For a given failure probability $\eta$, we can generate estimations $\hat{h}_{ij}$ and $\hat{\xi_i}$ with precision $\epsilon$ for all $i,\ j$ at the cost of
\begin{itemize}
    \item $\tilde{\mathcal{O}}\left(\epsilon^{-1} \log \left(\eta^{-1}\right)\right)$ total evolution time;
    \item $\tilde{\mathcal{O}}\left(\log\left(\epsilon^{-1}\right) \log \left(\eta^{-1}\right)\right)$ number of experiments;
    \item $\tilde{\mathcal{O}}\left(N \epsilon^{- 2}  \log \left(\eta^{-1}\right)\right)$ single-site random unitaries insertions.
\end{itemize}
Here we use the notation $\tilde{\mathcal{O}}$ to hide the higher order $\log$ terms for conciseness.
\end{theorem}
The readers can refer to \Cref{thm:main_detailed} for a formal version of this theorem. It is clear from the theorem that the Heisenberg limit can be obtained. 


\section{Method description}

We first discuss the simple case of single-site Fermionic learning, which can be solved by preparing proper initial states and performing an RPE-type algorithm for certain observables. This is followed by a discussion of the two-site case and many-site case, where we use the reshaping technique to decouple the Hamiltonian by inserting random unitaries during the simulation. This process will be demonstrated in detail in \Cref{sec:reshaping}.

\subsection{Single-site case}\label{sec:single-site}
For the single site case (i.e., when $N=1$), the Hamiltonian is reduced to 
\begin{equation}
H = \xi\nd\nup,
\end{equation}
where the only parameter to be learned is the coefficient $\xi$. It is clear that one of its eigenvalues is $\xi$ with the corresponding eigenstate $\kud$, and the vacuum state $\ket{-}$ is an eigenstate with eigenvalue $0$. It is worth noticing that extracting knowledge about $\xi$ from experiments requires more than simply initializing the system with its eigenstates since merely using the eigenstates as initial states of the simulation of $H$ will only result in a change in the global phase, which is unobservable. Therefore, the initial states of the experiments need to be a nontrivial linear combination of the eigenstates. For the sake of simplicity, we assume access to the state
\begin{equation}\label{eq:psi}
	\ket{\psi} = \sqhalf(\ket{-}+\kud).
\end{equation}
This state can be prepared as follows. Since $\ket{\psi}$ is a Fermionic state with 
even parity, it is a Gaussian state (see \cite{spee2018mode}). Then, one can apply the method in \cite{jiang2018quantum} for Fermionic Gaussian state preparation. 
We also assume we have access to its corresponding observable 
\begin{equation}\label{eq:o}
    O = \ket{\psi}\!\bra{\psi}.
\end{equation} 
This can also be fulfilled using the method in \cite{jiang2018quantum}, where the Gaussian state $\ket{\psi}$ is obtained by constructing a unitary $U$ that maps $\ket{-}$ to $\ket{\psi}$. With this unitary, we then have $\ket{\psi}\!\bra{\psi} = U\ket{-}\bra{-}U^\dagger$, which means $O$ can be performed by applying $U^\dagger$ to the state to be observed, and then observe it under the canonical basis. Using the initial state $\ket{\psi}$ and evolve under the Hamiltonian $H$ for time $t$, the expectation value of observable $O$ is 
\begin{equation}\label{eq:costxi}
    \braket{O}_{\psi,t} = \bra{\psi}e^{\ii Ht}O e^{-\ii Ht}\ket{\psi} =\frac{1+e^{\ii\xi t}}{2}\cdot\frac{1+e^{-\ii\xi t}}{2}= \half (1+\cos(t\xi)).
\end{equation}
We also assume access to
\begin{equation}\label{eq:tildepsi}
	\ket{\tilde{\psi}} = \sqhalf(\ket{-}+\ii\kud).
\end{equation}
This state also has even parity and is thus a Gaussian state, which can be prepared using the method in \cite{jiang2018quantum}.
By using $\ket{\tilde{\psi}}$ as the initial state and perform measurement using $O$, we get
\begin{equation}\label{eq:sintxi}
    \braket{O}_{\tilde{\psi},t} = \bra{\tilde{\psi}}e^{\ii Ht}O e^{-\ii Ht}\ket{\tilde{\psi}}=\frac{1-\ii e^{\ii\xi t}}{2}\cdot\frac{1+\ii e^{-\ii\xi t}}{2} = \half (1+\sin(t\xi)).
\end{equation}

Hence, by conducting multiple measurements, we can derive estimations for both $\cos(t\xi)$ and $\sin(t\xi)$. The accuracy is inversely proportional to the square of the number of measurements due to the Monte Carlo error estimation. Recall that $\kud$ is the eigenstate of $e^{\ii tH}$, and $\bra{\uparrow\downarrow}e^{\ii tH}\ket{\uparrow\downarrow} = e^{\ii tH} = \cos(t\xi) + \ii \sin(t\xi)$, we can conclude that the aforementioned estimations are equivalent to the estimating the value of $\bra{\uparrow\downarrow}e^{\ii tH}\ket{\uparrow\downarrow}$.

One can then apply robust phase estimation (RPE) \cites{kimmel2015robust, ni2023low} to obtain a noise-robust algorithm. Since the setting here is different to some extent, we reproduce the phase estimation results in \cite{ni2023low} for the sake of completeness.

\begin{algorithm}[ht]
	\caption{RPE algorithm}
	\label{alg:pruning}
	\begin{algorithmic}[1]
		\STATE{\textbf{Input:} $\eps$: target accuracy, $\eta$: upper bound of the failure probability. }
		\STATE{Let $J = \lceil\log_2(\frac{3}{\pi\eps})\rceil$ and calculate $N_s = 2\left\lceil 9\left(\log\frac{4}{\eta}+\log\left(\left\lceil\log_2\frac{3}{\pi\eps}\right\rceil+1\right)\right)\right\rceil$.}
		\STATE{$\theta_{-1} = 0$.}
		\FOR{$j = 0,1,\ldots,J$}
		\STATE{For $t = 2^j$, perform  $\frac{N_s}{2}$ times measurements for both \Cref{eq:costxi} and \Cref{eq:sintxi} to generate $X_j$ and $Y_j$ as estimations of $\cos(2^j\xi)$ and $\sin(2^j\xi)$ respectively. Let $Z_j = X_j+\ii Y_j$ be an estimation of $\bra{\uparrow\downarrow}e^{\ii 2^jH}\ket{\uparrow\downarrow} = e^{\ii 2^j \xi}$.}
		\STATE{Define a candidate set $S_j = \left\{\frac{2k\pi+\arg Z_j}{2^j} \right\}_{k=0,1,\ldots,2^j-1}$.}
		\STATE{$\theta_j = \argmin_{\theta\in S_j}\abs{\theta - \theta_{j-1}}$. }	
		\ENDFOR
		\STATE{\textbf{Output:} $\theta_J$ as an approximation to $\xi$. }
	\end{algorithmic}
\end{algorithm}

\begin{theorem}\label{thm:rpe}
Given target accuracy $\eps$ and admissible failure probability $\eta$, the output of \Cref{alg:pruning} satisfies
	\begin{equation}		\mathrm{Prob}\left(\abs{\theta_J - \xi} < \eps\right) > 1-\eta.
	\end{equation}
	In addition, the total evolution time of the Hamiltonian $H$ is 	\begin{equation}\mathcal{O}\left(\eps^{-1}\left(\log(\eta^{-1})+\log\log(\eps^{-1})\right)\right).
	\end{equation}    
\end{theorem}
\begin{proof}
    We will sketch the proof here, and the detailed proof can be found in \cite[Theorem 2]{ni2023low} for the case $\delta = 0$ and $U = e^{iH}$, which is the case corresponding to the exact eigenstate $\kud$. Notice that this algorithm aligns precisely with the RPE algorithm presented there, with the only distinction lying in the generation process of $Z_j$, which does not affect the proof. We also substitute all the $\eps$ with $\frac{\pi\eps}{3}$ to make the conventions align with this paper.

    \cmt{Using Hoeffding’s inequality,} the choice of $N_s$ guarantees that \begin{equation}\label{eq:zj}
	\abs{Z_j - \bra{\uparrow\downarrow}e^{\ii 2^jH}\ket{\uparrow\downarrow}} = \abs{Z_j - e^{\ii 2^j\xi}} < \frac{2}{3}<\frac{\sqrt{3}}{2} 
  \end{equation}
  holds for each $j = 0,1,\ldots,J$ with probability at least $1-\frac{\eta}{J+1}$, therefore it holds for all $j$'s with probability at least $1-\eta$. Once \Cref{eq:zj} holds for all $j$'s, then one can inductively prove that $\xi\in \left(\theta_j-\frac{\pi}{3\cdot 2^j},  \theta_j+\frac{\pi}{3\cdot 2^j}\right)$, where $\theta_j$ is the generated in the $j$-th step in \Cref{alg:pruning}. Therefore, $\abs{\xi-\theta_J}<\frac{\pi}{3\cdot 2^J}\le \eps$.
\end{proof}

\begin{remark}
In practice, the operators and initial states may be imperfect, which can still be dealt with by this RPE algorithm. Especially we want to highlight the scenarios where an $O(1)$ error is involved in the Hamiltonian simulation $e^{\ii Ht}$, resulting in the estimation of $e^{\ii 2^j \xi}$ with an additional $O(1)$ error. This case can be handled since there is still a gap between $\frac{\sqrt{3}}{2}$ and $\frac{2}{3}$ in \Cref{eq:zj}. This feature will be leveraged in the following sections when estimating the multiple-site Hamiltonian, where a reshaping error will be introduced.
\end{remark}

\subsection{Hamiltonian reshaping}\label{sec:reshaping}
When learning the multiple-site model, there are interaction terms with unknown coefficients in the Hamiltonian. Under this circumstance, it is unlikely to know the eigenstates of the Hamiltonian and take a similar approach as in the single-site case. The Hamiltonian reshaping technique aims to eliminate some of the coupling terms and reduce the objective Hamiltonian to several separate subsystems with known eigenstates by randomized insertion of unitaries during the Hamiltonian simulation. Note that we are not able to learn the coefficients of the eliminated terms by doing this reshaping. Nevertheless, through multiple experiments, each involving the elimination of different terms, one can collect estimations for all coefficients. 

Given a distribution $\DD$ of unitary matrices, the reshaped Hamiltonian is defined as 
\begin{equation}
    \heff = \expect_{U\sim \DD} U^\dagger H U.
\end{equation}
The proper choice of the distribution $\DD$ will make $\heff$ a simpler Hamiltonian than $H$ and facilitate the learning of the coefficients, which will be demonstrated in the following sections. 

Given this reshaped Hamiltonian, we use the following method, similar to qDRIFT \cite{Campbell2019random}, to approximately simulate $\heff$ as 
\begin{equation}\label{eq:qdrift}
    e^{-\ii t\heff}\approx e^{-\ii\tau U_r^\dagger H U_r}\cdots e^{-\ii\tau U_1^\dagger H U_1},
\end{equation}
where $r$ is an integer and $\tau = t/r$ is the time stepsize. Each $U_j$ is randomly sampled from the distribution $\DD$. By leveraging the identity
\begin{equation}
    e^{-\ii\tau U_j^\dagger H U_j} = U_j^\dagger e^{-\ii\tau  H}U_j,
\end{equation}
it is clear that each term in \Cref{eq:qdrift} can be implemented in experiments by inserting unitaries when simulating the original Hamiltonian $H$.

For a fixed $t$, if $r\to \infty$, then the approximation in \Cref{eq:qdrift} will be exact. In practice, we can choose an appropriate $r$ to get the desired accuracy. We refer to this error caused by using a finite $r$ as the reshaping error. As an analogy of the first-order Trotter formula error estimation, we would expect $r = O(t^2)$, which is a quadratic dependence on the simulation time. The rigorous discussions of the reshaping error and the choice of $r$ will be given in the following sections.

\subsection{Two-site case}\label{sec:two-site}
In this section, we address the problem of learning a two-site model, in which the Hamiltonian can be written as 
\begin{equation}\label{eq:two-site H}
    H = -h_{12} (\coud\ctu+\ctud\cou+\codd\ctd+\ctdd\cod)+\sum_{i=1,2}\xi_i \niu\nid,
\end{equation}
and the goal is to learn the coefficients $\xi_1$, $\xi_2$, and $h_{12}$.
	
First, we focus on learning the coupling coefficient $h_{12}$. If we prepare initial states that only have spin $\uparrow$, then the terms containing spin $\downarrow$ in the Hamiltonian will be eliminated. Under this setting, the system is reduced to a two-mode system, and the effective Hamiltonian becomes
\begin{equation}\label{eq:two_init}
	H_{\mathrm{eff}} = -h_{12} (\coud\ctu+\ctud\cou).
\end{equation}
Notice that the states 
$$\sqhalf(\ket{\uparrow}_1\pm\ket{\uparrow}_2)$$ 
are the eigenstate of $\coud\ctu+\ctud\cou$ with eigenvalues $\pm 1$ respectively. Then, the problem is again a phase estimation problem, and one can apply RPE. Similar to the single-site case, one needs to prepare the superposition of the eigenstates
\begin{equation}\label{eq:phi}
    \ket{\phi} = \sqhalf\left(\sqhalf(\ket{\uparrow}_1+\ket{\uparrow}_2)+\sqhalf(\ket{\uparrow}_1-\ket{\uparrow}_2)\right) = \ket{\uparrow}_1
\end{equation}
and
\begin{equation}\label{eq:tildephi}
    \ket{\tilde{\phi}} = \sqhalf\left(\sqhalf(\ket{\uparrow}_1+\ket{\uparrow}_2)+\ii\sqhalf(\ket{\uparrow}_1-\ket{\uparrow}_2)\right) = \frac{1+\ii}{2}\ket{\uparrow}_1 + \frac{1-\ii}{2}\ket{\uparrow}_2
\end{equation}
as the initial states of the Hamiltonian simulation of $H$. These states are again Gaussian Fermionic states since they are two-mode states with odd parity (one can refer to \cite{spee2018mode}), which can thus be prepared by the algorithms from \cite{jiang2018quantum}. We also assume access to the observable 
\begin{equation}\label{eq:o2}
    O^{(2)} = \ket{\phi}\!\bra{\phi}.
\end{equation}
\cmt{
\begin{remark}
    This method can be easily extended to the case that the coefficients of spin up and down are different, which is 
    $$H = -h_{12\uparrow} (\coud\ctu+\ctud\cou)-h_{12\downarrow}(\codd\ctd+\ctdd\cod)+\sum_{i=1,2}\xi_i \niu\nid.$$
    One can use the eigenstates    $$\sqhalf(\ket{\uparrow}_1\pm\ket{\uparrow}_2)$$
    to learn $h_{12\uparrow}$, and eigenstates   $$\sqhalf(\ket{\downarrow}_1\pm\ket{\downarrow}_2)$$
    to learn $h_{12\downarrow}$.
\end{remark}
\begin{remark}
    In Hamiltonian \eqref{eq:two-site H}, we can also use the symmetrized eigenstates    $$\sqhalf(\ket{\uparrow}_1+\ket{\downarrow}_1)\pm(\ket{\uparrow}_2+\ket{\downarrow}_2)$$
    to learn $h_{12}$, which would be preferable on certain experiment platforms.
\end{remark}
}
Next, we aim to learn the coefficients $\xi_i$. In order to do this, one can insert a random unitary matrix of form
\begin{equation}
   U = e^{-\ii\theta(n_{1\uparrow}+n_{1\downarrow})}, \quad \theta\sim\mathcal{U}([0,2\pi])
\end{equation}
where $\mathcal{U}([0,2\pi])$ denotes the uniform distribution on $[0,2\pi]$. The effective Hamiltonian is then
\begin{equation}\label{eq:eliminate_cross_term}
\begin{aligned}
H_{\mathrm{eff}} &= \frac{1}{2\pi}\int_{0}^{2\pi}e^{\ii\theta(n_{1\uparrow}+n_{1\downarrow})}He^{-\ii\theta(n_{1\uparrow}+n_{1\downarrow})}\d \theta\\
&=\sum_{i=1,2}\xi_i \niu\nid-\frac{h_{12}}{2\pi}\int_{0}^{2\pi}e^{\ii\theta(n_{1\uparrow}+n_{1\downarrow})}(\coud\ctu+\ctud\cou+\codd\ctd+\ctdd\cod)e^{-\ii\theta(n_{1\uparrow}+n_{1\downarrow})}\d \theta\\
&=\sum_{i=1,2}\xi_i \niu\nid-\frac{h_{12}}{2\pi}\int_{0}^{2\pi}(e^{\ii\theta}\coud\ctu+e^{-\ii\theta}\ctud\cou+e^{\ii\theta}\codd\ctd+e^{-\ii\theta}\ctdd\cod)\d \theta\\
&=\sum_{i=1,2}\xi_i \niu\nid,
\end{aligned}
\end{equation}
where we have used the identities
\[
e^{\ii \theta n_\sigma}c_{\sigma}e^{-\ii \theta n_\sigma} = e^{-\ii\theta}c_{\sigma},~ e^{\ii \theta n_\sigma}c_{\sigma}^\dagger e^{-\ii \theta n_\sigma} = e^{\ii\theta}c_{\sigma}^\dagger, \quad \sigma \in\{ \uparrow, \downarrow\},
\]
which can be proved by applying them to the basis $\{\ket{-},\ket{\uparrow},\ket{\downarrow},\kud\}$. This effective Hamiltonian is decoupled, and we can apply the method for the single-site case. The reshaping error estimation will be postponed to the end of \Cref{sec:many-site}, as it can be combined into the discussion of the many-site reshaping error.

	
\subsection{Multiple-site case}\label{sec:many-site}
To learn the coefficients of a many-site Hubbard Hamiltonian, we adopt a divide-and-conquer approach, which involves the Hamiltonian reshaping technique as well. To illustrate the main idea, we first consider a one-dimensional chain of $N$ sites, characterized by the Hamiltonian
\begin{equation}
    H = -\sum_{|i-j|=1, \sigma\in\{\uparrow,\downarrow\}} h_{ij}\cisd\cjs+\sum_{i=1}^N\xi_i \niu\nid,\quad (h_{ij} = h_{ji}).
\end{equation}
One can insert random unitaries
\[
U = \prod_{i = 3,6,9,\ldots}e^{-\ii\theta_i(\niu+\nid)},
\]
where $\theta_i$'s are random numbers uniformly drawn from $[0,2\pi]$. This will eliminate the terms hopping from and to sites $i=3,6,9,\ldots$, which can be derived using the same calculation as in \Cref{eq:eliminate_cross_term}. The resulting effective Hamiltonian will, therefore, be
\begin{equation}
    \heff = H_{12}+H_{45} + H_{78} + \cdots,
\end{equation}
where $H_{j,j+1}$ is the Hamiltonian
\begin{equation}
    H_{j,j+1} = -h_{j,j+1} \sum_{\sigma\in\{\uparrow,\downarrow\}}(c_{j,\sigma}^\dagger c_{j+1,\sigma} + c_{j+1,\sigma}^\dagger c_{j,\sigma})+\xi_j n_{j,\uparrow}n_{j,\downarrow} + \xi_{j+1} n_{j+1,\uparrow}n_{j+1,\downarrow}
\end{equation}
for $j = 1,4,7,\ldots$, which is the restriction of $H$ on edge $(j,j+1)$. This means the effective Hamiltonian $\heff$ is now decoupled into several two-site Hamiltonians, and we can use the method in \Cref{sec:two-site} to learn the coefficients $\xi_{j}$, $\xi_{j+1}$, and $h_{j,j+1}$ for $j = 1,4,7,\ldots$. Using a similar method as above, if we insert the random unitaries 
\begin{equation*}
    U = \prod_{i = 1,4,7,\ldots}e^{-\ii\theta_i(\niu+\nid)} \quad \text{or}\quad U = \prod_{i = 2,5,8,\ldots}e^{-\ii\theta_i(\niu+\nid)}
\end{equation*}
instead, then the effective Hamiltonian will be 
\begin{equation}
    \heff = H_{23}+H_{56} + H_{89} + \cdots\quad \text{or}\quad \heff = \xi_1n_{1\uparrow}n_{1\downarrow}+H_{34}+H_{67} + \cdots,
\end{equation}    
respectively. Combining the results of these three experiments, we can learn all the coefficients $\xi_i$ and $h_{i,i+1}$.

Recall that, in the general setting, we assume the Fermionic sites are located on a graph $G = (V, E)$. Each vertex $i\in V$ denotes a Fermionic site, and each edge in $E$ denotes a pair of sites that have interactions between them. The Hamiltonian is 
\[
H = -\!\!\!\!\sum_{i\sim j, \sigma\in\{\uparrow,\downarrow\}}\!\! h_{ij}\cisd\cjs+\sum_{i}\xi_i \niu\nid,\quad (h_{ij} = h_{ji}),
\]
where $i\sim j$ means edge $(i,j)\in E$. Since we only consider the local interaction in the Hubbard model, we may assume the maximum degree of graph $G$ is \begin{equation}\label{eq:local interaction}
	\mathfrak{d} = \max_{i\in V}\deg (i) = \mathcal{O}(1).
\end{equation}

Our strategy is to classify the edges of the graph with a small number of colors, and the edges of the same color can be viewed as separate two-site Hubbard models. When learning the coefficients $h_{ij}$ (and the corresponding $\xi_i$ and $\xi_j$) of a certain color, we only need to block out the interference from other sites using the Hamiltonian reshaping technique mentioned above. After this reshaping, these separate two-site Hubbard models can be dealt with by the method in \Cref{sec:two-site}. Next, we will specify how to do this coloring. 

Consider another graph $\GG = (E,\EE)$, whose vertex set is exactly the edge set of graph $G$. The edge set $\EE$ is defined as follows. For $C$ and $C'$ in $E$, there is an edge $(C, C')$ in $\EE$ if one of the following two conditions is satisfied.
\begin{itemize}
    \item $C$ and $C'$ share a vertex in graph $G$.
    \item There exists a $C''$ in $E$ such that $C$ and $C''$ share a vertex in $G$, and $C'$ and $C''$ also share a vertex in $G$.
\end{itemize}
Under this definition, one can see that the maximal degree of $\GG$ is at most $4\mathfrak{d}^2$ since each edge $C$ can only have at most $2\mathfrak{d}$ adjacent edges in graph $G$. Therefore, the vertex set of graph $\GG$ can be colored using $\chi = 4\mathfrak{d}^2+1 = O(1)$ different colors, such that for every $(C,C')\in \EE$, the colors of $C$ and $C'$ are different. Such coloring can be obtained by a simple greedy algorithm that colors the vertices one by one and introduces a new color only if all existing color is repeated with the already-colored neighbor vertices. Therefore, we can partition the set $E$ as
\begin{equation}
    E = \bigsqcup_{c=1}^{\chi}E_c,
\end{equation}
where $E_c$ is the subset with color $c$. We also define $V_c$ as all the vertices that are contained in some edges in $E_c$. For each $C\in E_c$, denote its two vertices as $k^C_1$ and $k^C_2$. Let \begin{equation}\label{eq:vc12}
    V_{c_1}=\{k^C_1:\ C\in E_c\}\text{, and }V_{c_2}=\{k^C_2:\ C\in E_c\}.
\end{equation} 
Therefore, we have $V_c = V_{c_1}\cup V_{c_2}$. For illustration, we point out that in the Fermionic chain example above, we have $\chi = 3$, and $E_c$ contains all the edges of the form $(3k+c,3k+c+1)$, where $k\ge 0, c = 1,2,3$. The strategy of the general case is similar. When we want to learn the coefficients associated with a certain color $c$, we can apply the random unitaries
\begin{equation}\label{eq:random unitary}
    U = \prod_{i \in V\backslash V_c}e^{-\ii\theta_i(\niu+\nid)}
\end{equation}
to reshape the Hamiltonian. Due to the coloring rules, one can check that the resulting effective Hamiltonian is
\begin{equation}
    \heff = \sum_{C\in E_c} \heff^C 
\end{equation}
where $\heff^C$ is the two-site restriction of $H$ on edge $C$. To control the reshaping error, we have the following theorem. It is noteworthy that the reshaping error bound is independent of system size $N$. This can be understood intuitively, as the underlying graph $G$ has a bounded degree, making it challenging for most of the sites to significantly influence the subsystem $C$.
\begin{theorem}\label{thm:reshaping error}
    We assume random unitary operators $U_l, 1 \leq l \leq r$, are generated independently and are identically distributed as $U$, which satisfies
$$
\mathbb{E}[U^\dagger H U]=H_{\mathrm{eff}}^C+H_{\mathrm{env}},
$$
where $H_{\mathrm{eff}}^C$ is supported on a subsystem $C(|C|=\mathcal{O}(1))$ and $H_{\mathrm{env}}$ is supported on the rest of the system. Then
\begin{equation}\label{eq:reshaping error operator norm}
    \left\|\mathbb{E}\left[\prod_{1 \leq l \leq r}^{\rightarrow}\left(U_l^\dagger e^{\ii H \tau} U_l\right)\left(O_C \otimes I\right) \prod_{1 \leq l \leq r}^{\leftarrow}\left(U_l^\dagger e^{-\ii H \tau} U_l\right)\right] - e^{\ii H_{\mathrm{eff}}^C t} O_C e^{-\ii H_{\mathrm{eff}}^C t} \otimes I\right\|=\mathcal{O}\left(t^2 / r\right)
\end{equation}
for any $O_C$ supported on $C$ satisfying $\left\|O_C\right\| \leq 1$. In particular, the constant in $\mathcal{O}\left(t^2 / r\right)$ does not depend on the system size $N$.
\end{theorem}

This theorem is slightly modified from the \cite[Appendix D, Theorem 16]{huang2023learning}, and the proof is almost the same. Though \cite{huang2023learning} only discussed the case of each $U_l$ being the tensor product of Pauli matrices, the method of proof there applies to general unitary matrices. 

Now, we have collected the ingredients for analyzing the estimation of coefficients $h_{ij}$'s. However, when estimating the coefficients $\xi_i$'s, another reshaping is needed as described in \Cref{sec:two-site}. Notice that the two reshaping processes are inserting unitaries with non-intersect supports. Therefore, they can be combined into a single reshaping process, i.e., inserting random unitaries sampled from the product distribution of the two distributions. To be more concrete, assume that the task is to learn $\xi_{k_1^0}$, where $k_1^0\in V_{c_1}$ is an index belonging to an edge $C_0 = (k_1^0,k_2^0)$ of color $c_1$. Then, the two steps of reshaping are equivalent to directly inserting random unitaries of the form
\begin{equation}\label{eq:reshape k1}
	U = \prod_{i \in (V\backslash V_{c_1})}e^{-\ii\theta_i(\niu+\nid)}.
\end{equation}
After reshaping, the effective Hamiltonian becomes
\begin{equation*}
	\heff = \heff^{\{k_1^0\}} + H_{\mathrm{env}}.
\end{equation*}
By viewing $\{k_1^0\}$ as the subsystem $C$ in \Cref{thm:reshaping error}, we conclude the same reshaping error estimation applies to $\xi_{k_1^0}$. By inserting random unitaries
\begin{equation}\label{eq:reshape k2}
	U = \prod_{i \in (V\backslash V_{c_2})}e^{-\ii\theta_i(\niu+\nid)}.
\end{equation}
we can also learn $\xi_{k_2^0}$. 

The algorithm for learning the Hamiltonian \Cref{eq:hamiltonian} is outlined in \Cref{alg:learning}. \cmt{We summarize the necessary operations on an experimental platform:
\begin{itemize}
    \item preparing the vacuum state $\ket{-}$;
    \item the fermion Gaussian state on at most two sites, together with its preparing operator \cite{jiang2018quantum, spee2018mode};
    \item measurement of spin-orbital occupation;
    \item implementing $e^{\ii\theta n_{\uparrow}}$ and $e^{\ii\theta n_{\downarrow}}$ for a single site, which is less demanding than in \cite{bravyi2002fermionic}.
\end{itemize}}

The following theorem summarizes the complexity bounds for this algorithm under \Cref{eq:local interaction}.

\begin{algorithm}[ht]
	\caption{Fermionic Hamiltonian learning for Hubbard models}
	\label{alg:learning}
	\begin{algorithmic}[1]
		\STATE{\textbf{Input:} The Hamiltonian $H$ of the form \Cref{eq:Hubbard}, and the corresponding graph $G = (V,E)$.}
		\STATE{Generate the coloring partition $E = \bigsqcup_{c=1}^{\chi}E_c$.}
            \STATE{Calculate $J$ and $N_s$ according to \Cref{alg:pruning}.}
		\FOR{$c = 1,2,\ldots, \chi$}
            \STATE{Determine $V_{c_1}$, $V_{c_2}$, and $V_c$ according to \Cref{eq:vc12}.}
            \FOR{$j = 0,1,\ldots, J$}
            \STATE{Prepare initial state $\ket{\Phi_c} = \bigwedge_{C\in E_c}\ket{\phi_C}$, where $\ket{\phi_C}$ is the state defined in \Cref{eq:phi} corresponding to edge $C$.}
		\STATE{Let $t = 2^j$. Insert the random unitaries of the form \Cref{eq:random unitary} every $\tau$ time while evolving the Hamiltonian $H$, where the time step $\tau$ is chosen such that the reshaping error (\ref{eq:reshaping error operator norm}) is no more than $(\frac{\sqrt{3}}{2}-\frac{2}{3})/4$.}
            \STATE{Perform measurement using $O_c = \bigotimes_{C\in E_c}O^{(2)}_C$, where $O^{(2)}_C$ is the observable defined in \Cref{eq:o2}.} \STATE{Perform this observation for $N_s/2$ times and take the average to get $X_j^C$ as an estimation of $\cos(2^j h_{k^C_1 k^C_2})$. (This is simultaneously for all $C\in E_c$.)}
            \STATE{Similarly, use initial state $\ket{\tilde{\Phi}_c} = \bigwedge_{C\in E_c}\ket{\tilde{\phi}_C}$ to get $Y_j^C$ as an estimate of $\sin(2^j h_{k^C_1 k^C_2})$. }
            \STATE{Use RPE to get $\hat{h}_{k^C_1 k^C_2}$ as an estimation of $h_{k^C_1 k^C_2}$. }
        \FOR{$m = 1,2$}
            \STATE{Prepare initial state $\ket{\Psi_{cm}} = \bigwedge_{i\in V_{cm}}\ket{\psi_i}$, where $\ket{\psi_i}$ is the state defined in \Cref{eq:psi} corresponding to site $i$.}
		\STATE{Let $t = 2^j$. Insert the random unitaries of the form \Cref{eq:reshape k1} or \Cref{eq:reshape k2} every $\tau$ time while evolving the Hamiltonian $H$, where the time step $\tau$ is chosen such that the reshaping error (\ref{eq:reshaping error operator norm}) is no more than $(\frac{\sqrt{3}}{2}-\frac{2}{3})/4$.}
            \STATE{Perform measurement using $O_{cm} = \bigotimes_{i\in V_{cm}}O_i$, where $O_i$ is the observable defined in \Cref{eq:o}.} \STATE{Perform this observation for $N_s/2$ times and take the average to get $X_j^{\{i\}}$ as an estimation of $\cos(2^j \xi_i)$ for $i\in V_{cm}$.}
            \STATE{Similarly, use initial state $\ket{\tilde{\Psi}_{cm}} = \bigwedge_{i\in V_{cm}}\ket{\tilde{\psi}_i}$ to get $Y_j^{\{i\}}$ as an estimate of $\sin(2^j \xi_i)$.}
            \STATE{Use RPE to get $\hat{\xi_i}$ as an estimation of $\xi_i$. }
            \ENDFOR
            \ENDFOR
		\ENDFOR
		\STATE{\textbf{Output:} $\hat{h}_{ij}$ and $\hat{\xi_i}$. }
	\end{algorithmic}
\end{algorithm}

\begin{theorem}\label{thm:main_detailed}
	Assume $H$ is a Hamiltonian in the form \Cref{eq:hamiltonian} satisfying \Cref{eq:local interaction}. For a given failure probability $\eta$, we can generate estimations $\hat{h}_{ij}$ and $\hat{\xi_i}$ from \Cref{alg:learning} such that 
	\begin{equation}\label{eq:acc}
		\mathrm{Prob}(|\hat{h}_{ij}-h_{ij}|<\eps) > 1-\eta \quad\text{and}\quad \mathrm{Prob}(|\hat{\xi_{i}}-\xi_{i}|<\eps) > 1-\eta
	\end{equation}
for all $i,\ j$ at the cost of
\begin{itemize}
    \item $\tilde{\mathcal{O}}\left(\epsilon^{-1} \log \left(\eta^{-1}\right)\right)$ total evolution time;
    \item $\tilde{\mathcal{O}}\left(\log\left(\epsilon^{-1}\right) \log \left(\eta^{-1}\right)\right)$ number of experiments;
    \item $\tilde{\mathcal{O}}\left(N \epsilon^{- 2}  \log \left(\eta^{-1}\right)\right)$ single-site random unitaries insertions.
\end{itemize}
\end{theorem}
\begin{proof}
    We first prove \Cref{eq:acc}. For a site $i$, we may assume $i\in V_{c_1}$ without loss of generality. Since we chose the time step $\tau$ such that the reshaping error is no more than $(\frac{\sqrt{3}}{2}-\frac{2}{3})/4$, which means the difference between
    \begin{equation*}    
    \mathbb{E}\half (1+X_j^{\{i\}}) = \mathbb{E}\bra{\Psi_{c_1}}\left[\prod_{1 \leq l \leq r}^{\rightarrow}\left(U_l^\dagger e^{\ii H \tau} U_l\right)\left(O_i \otimes I\right) \prod_{1 \leq l \leq r}^{\leftarrow}\left(U_l^\dagger e^{-\ii H \tau} U_l\right)\right]\ket{\Psi_{c_1}}
    \end{equation*}
    and 
    \begin{equation*} \bra{\Psi_{c_1}}e^{\ii H_{\mathrm{eff}}^i t} O_i e^{-\ii H_{\mathrm{eff}}^i t} \otimes I\ket{\Psi_{c_1}} = \half (1+\cos(t\xi_i))
    \end{equation*}
    is at most $(\frac{\sqrt{3}}{2}-\frac{2}{3})/4$, where $t = 2^j$. Therefore, 
    \begin{equation*}
        \abs{\mathbb{E} X_j^{\{i\}} - \cos(t\xi_i)}\le(\frac{\sqrt{3}}{2}-\frac{2}{3})/2.
    \end{equation*}
    We also have the same estimation for $\mathbb{E} Y_j^{\{i\}}$, therefore we have
    \begin{equation}
        \abs{\mathbb{E} (X_j^{\{i\}}+\ii Y_j^{\{i\}}) - e^{t\xi_i}}\le\frac{\sqrt{3}}{2}-\frac{2}{3}.
    \end{equation}
    \cmt{We also have to consider the Monte Carlo error. It can be bounded by $\frac{2}{3}$ using Hoeffding's inequality as the same reason in \Cref{eq:zj}, which is
    \begin{equation}
        \abs{(X_j^{\{i\}}+\ii Y_j^{\{i\}}) - \mathbb{E} (X_j^{\{i\}}+\ii Y_j^{\{i\}})}\le\frac{2}{3},
    \end{equation}
    for all $j=0,1,\ldots,J+1$ with probability $1-\eta$.}
    Therefore, we conclude that with probability $1-\eta$, we have 
    \begin{equation}
        \abs{(X_j^{\{i\}}+\ii Y_j^{\{i\}}) - e^{t\xi_i}}\le\frac{\sqrt{3}}{2},
    \end{equation}
    for all $j=0,1,\ldots,J+1$. These bounds guarantee the RPE algorithm to give an $\eps$ accurate estimation of $\xi_i$.

    It remains to analyze the cost. Since the number of colors $\chi=\mathcal{O}(1)$ as demonstrated above, the number of experiments is of order $J\cdot N_s = \tilde{\mathcal{O}}\left(\log\left(\epsilon^{-1}\right) \log \left(\eta^{-1}\right)\right).$ In each experiment, the longest evolution time of $H$ is $2^J = \mathcal{O}(\eps^{-1})$, which makes the total evolution time $\mathcal{O}(J N_s\eps^{-1}) = \tilde{\mathcal{O}}\left(\epsilon^{-1} \log \left(\eta^{-1}\right)\right).$ For each experiment, we have to insert the \cmt{$r = \mathcal{O}(t^2)$} random unitaries of the form \Cref{eq:random unitary}, \Cref{eq:reshape k1}, or \Cref{eq:reshape k1}, according to \Cref{eq:reshaping error operator norm} as we have to achieve $(\frac{\sqrt{3}}{2}-\frac{2}{3})/4 = \mathcal{O}(1)$ accuracy. \cmt{At the maximal time $t=2^J$, we need $r = \mathcal{O}((2^J)^2) = \mathcal{O}(\eps^{-2})$ random unitaries,} and each random unitary contains $\mathcal{O}(N)$ single-site unitaries. Therefore, we conclude that the total number of single unitary insertions is 
    \begin{equation*}
        \mathcal{O}(\eps^{-2}NN_sJ) = \tilde{\mathcal{O}}\left(N \epsilon^{- 2}  \log \left(\eta^{-1}\right)\right).
    \end{equation*}
\end{proof}

\section{Conclusion and discussion}
This work proposed a method for learning an unknown Fermionic Hamiltonian of Hubbard models. Our method achieves the $\tilde{\mathcal{O}}(\eps^{-1})$ scaling of total evolution time, which matches the Heisenberg limit. Notably, it relies only on elementary manipulations applicable to single-site or two-site Fermionic systems, which is desirable in the experimental platforms.

On the theoretical aspect, several open problems are worth further exploration. One problem is the extension of our method to deal with more general Fermionic Hamiltonians beyond the Hubbard model. For example, whether it is feasible to learn Hamiltonians featuring long-range interactions within a total evolution time of $\tilde{\mathcal{O}}(\eps^{-1})$ while simultaneously minimizing dependence on the system size $N$ is an interesting question to consider. Another problem is whether the quadratic $\tilde{\mathcal{O}}(\eps^{-2})$ term in the number of single-site random unitary insertion can be reduced by substituting the Hamiltonian reshaping method used here by some analogs of higher order trotter formulas, similar to the considerations in \cite{huang2023learning} under the spin setting.

Regarding the experimental aspects, several problems can be further investigated. For instance, the states and observables we discussed may not be perfectly prepared due to the limitations of experimental platforms. However, an $\mathcal{O}(1)$ amount of noise is allowed, which is a notable feature of the RPE algorithm. One may also be able to modify our method to accommodate initial states with a certain degree of randomness while still attaining the desired results. Another possible improvement is to replace the discrete random unitary insertion with some continuous random unitary evolution, which may be more desirable to experimental implementations.

\section*{Statements and Declarations}
\paragraph{\textbf{Competing Interests:}} The authors declare that they have no competing interests.

\paragraph{\textbf{Funding:}}  The work of L.Y. is partially supported by the National Science Foundation under awards DMS-2011699 and DMS-2208163.

\paragraph{\textbf{Acknowledgements:}}  We thank Yu Tong for the helpful comments on the initial draft of this work.

\bibliographystyle{unsrtnat}
\bibliography{ref}
	
\end{document}